\documentclass[11pt]{article}
\usepackage{amsfonts}
\usepackage{amsmath,amsthm}

\usepackage{mathrsfs}
\usepackage{fancyhdr}
\usepackage{latexsym}
\usepackage{amssymb}
\usepackage{color}
\usepackage{cite}
\usepackage{hyperref}
\newtheorem{Lemma}{Lemma}[section]

\newtheorem{Proposition}{Proposition}[section]
\numberwithin{equation}{section}

\title{On the Lattice Potential KP Equation}

\author{Cewen Cao$^1$, ~ Xiaoxue Xu$^1$, ~ Da-jun Zhang$^2$\footnote{Corresponding author. E-mail: djzhang@staff.shu.edu.cn}
\\
{\small $^1$School of Mathematics and Statistics, Zhengzhou University, Zhengzhou Henan 450001, P.R. China} \\
{\small $^2$Department of Mathematics, Shanghai University, Shanghai 200444, P.R. China}\\
{\small E-mail:\ \ cwcao@zzu.edu.cn, xiaoxuexu@zzu.edu.cn, djzhang@staff.shu.edu.cn}
}

\begin{document}

\maketitle

\begin{abstract}
The paper presents an approach to derive finite genus solutions to
the lattice potential Kadomtsev-Petviashvili (lpKP) equation introduced by F.W. Nijhoff, et al.
This equation is rederived from compatible conditions of three replicas of the discrete ZS-AKNS spectral problem,
which is a Darboux transformation of the continuous ZS-AKNS spectral problem.
With the help of these links and by means of the so called nonlinearization technique and Liouville platform,
finite genus solutions of the lpKP equation are derived.
Semi-discrete potential KP equations with one and two discrete arguments, respectively, are also discussed.

\end{abstract}

\section{Introduction }\label{sec-1}

Discrete integrable systems and the problem of integrable discretization of given soliton equations
have attracted more and more attention in recent years \cite{12-GraKT-Spr-2004,13-HieJN-Cam-2016,19-Sur-Bir-2003}.
The main purpose of this paper is to investigate the lattice
potential Kadomtsev-Petviashvili (lpKP) equation
\begin{align}
\Xi_{(\beta_1,\beta_2,\beta_3)}^{\text{lpKP}}\equiv\, & (\beta_1-\tilde{W})(\beta_2-\beta_3+\hat{\tilde{W}}-\bar{\tilde{W}})
+(\beta_2-\bar{W})(\beta_3-\beta_1+\tilde{\bar{W}}-\hat{\bar{W}}) \nonumber\\
& +(\beta_3-\hat{W})(\beta_1-\beta_2+\bar{\hat{W}}-\tilde{\hat{W}})=0,\label{1.1}
\end{align}
and present an approach to construct finite genus solutions to 3D integrable lattice equations.
This equation is first discovered by Nijhoff, Capel, Wiersma and Quispel  by using the
B\"{a}cklund transformation approach, and later derived through an
analysis of the Cauchy matrix \cite{17-NijCWQ-PLA-1984,13-HieJN-Cam-2016}.

To build the integrability of the lpKP equation \eqref{1.1} and calculate its finite genus solutions,
we will introduce Lax triads from the ZS-AKNS spectral problem.
Compatibility of these triads, respectively, give rise to  the lattice potential KP equations with 3, 2 and 1 discrete arguments,
as (see Section \ref{sec-2} for the derivation)
%
\begin{align}
\Xi^{(0,3)}\equiv
&\frac{1}{2}[(\tilde{W}-\bar{W})\bar{\tilde{W}}
+(\bar{W}-\hat{W})\hat{\bar{W}}+(\hat{W}-\tilde{W})\tilde{\hat{W}}]\nonumber\\
&+\gamma_1(\bar{\tilde{W}}-\hat{\tilde{W}}-\bar{W}+\hat{W})
+\gamma_2(\hat{\bar{W}}-\tilde{\bar{W}}-\hat{W}+\tilde{W})\nonumber\\
&+\gamma_3(\tilde{\hat{W}}-\bar{\hat{W}}-\tilde{W}+\bar{W})=0,\label{1.11}
\end{align}
\begin{equation}
\Xi^{(1,2)}\equiv(\tilde{W}-\bar{W})_x-[\frac{1}{2}(\tilde{W}-\bar{W})+\gamma_1-\gamma_2]
(\bar{\tilde{W}}-\tilde{W}-\bar{W}+W)=0,
\label{1.12}
\end{equation}
\begin{equation}
\Xi^{(2,1)}\equiv(\tilde{W}-W)_y-[(\tilde{W}+W)_x+2\gamma_1(\tilde{W}-W)+\frac{1}{2}(\tilde{W}-W)^2]_x=0.
\label{1.13}
\end{equation}
Note that (\ref{1.11}) is equivalent to the lpKP equation (\ref{1.1}) with
$\beta_k=-2\gamma_k,\, k=1,2,3$, as
$
\Xi_{(-2\gamma_1,-2\gamma_2,-2\gamma_3)}^{\text{lpKP}}=2\Xi^{(0,3)}$.
It also turns out that  all these equations have the same potential KP (pKP) equation,
\begin{equation}
\Xi^{(3,0)}\equiv
W_{xt}-\frac{1}{4}(W_{xxx}+3W^2_x)_x-\frac{3}{4}W_{yy}=0 ,
\label{1.3}
\end{equation}
as their continuum limits (see Proposition \ref{P-2.5}).

The method of finite-gap integration originated in solving the periodic initial problem
of the Korteweg-de Vries (KdV) equation (cf.\cite{Mat-RTRSA-2008} and the references therein).
Recently, an approach to deriving finite genus solutions for 2D discrete integrable systems,
the lattice potential KdV  equation \cite{7-CaoX-JPA-2012} and
the lattice nonlinear Schr\"{o}dinger (lNLS)  model \cite{8-CaoZ-JPA-2012}  were developed.
In this paper, just as in the 2D case, explicit analytic solutions of the lattice pKP equations (\ref{1.11},\ref{1.12},\ref{1.13}),
together with the pKP equation \eqref{1.3},
will be calculated by means of  the finite-dimensional integrable flows
of continuous and discrete types, i.e. Hamiltonian phase flows and integrable symplectic maps.
These flows are constructed through nonlinearization of the continuous and discrete spectral problems
 (see Section \ref{sec-3},\ref{sec-4}). It is surprising
that they share same Liouville integrals, same Lax matrix
$L(\lambda;p,q)$ and  same algebraic curve $\mathcal{R}$.
Thus the calculations can be done on the same Liouville platform.
The Abel-Jacobi variable $\vec{\phi}$ in the Jacobian variety
$J(\mathcal{R})$ straightens out both the $H_j$- and the
$S_{\gamma_k}$-flow with the velocities $\vec{\Omega}_j$ and
$\vec{\Omega}_{\gamma_k}$, respectively.
As a result, we have a clear
evolution picture for the lattice pKP equations as well as for the pKP equation, as the following,
\begin{align*}
&\Xi^{(0,3)}:\quad \vec{\phi}\equiv \vec{\phi}_0+m_1\vec{\Omega}_{\gamma_1}+m_2\vec{\Omega}_{\gamma_2}+m_3\vec{\Omega}_{\gamma_3},\quad (\text{mod}\mathscr T), \\
&\Xi^{(1,2)}:\quad \vec{\phi}\equiv \vec{\phi}_0+x\vec{\Omega}_1+m_1\vec{\Omega}_{\gamma_1}+m_2\vec{\Omega}_{\gamma_2},
\quad (\text{mod}\mathscr T), \\
&\Xi^{(2,1)}:\quad \vec{\phi}\equiv \vec{\phi}_0+x\vec{\Omega}_1+y\vec{\Omega}_2+m_1\vec{\Omega}_{\gamma_1},
\quad (\text{mod}\mathscr T), \\
&\Xi^{(3,0)}:\quad \vec{\phi}\equiv
\vec{\phi}_0+x\vec{\Omega}_1+y\vec{\Omega}_2+t\vec{\Omega}_3,\quad(\text{mod}\mathscr T),
\end{align*}
which will  provide the basic part of the explicit analytic solutions (see
Section \ref{sec-5},\ref{sec-6}).

The paper is organized as follows.
Section \ref{sec-2} shows that how the lattice pKP equations  (\ref{1.11}), (\ref{1.12}) and (\ref{1.13})
arise from their Lax triads. Continuum limits of these lattice pKP equations
give rise to the same pKP equation.
In Section \ref{sec-3}, a finite-dimensional integrable Hamiltonian system related to the  ZS-AKNS spectral problem
 is introduced to provide  integrals,  spectral curve and Abel-Jacobi variables.
In Section \ref{sec-4}  we construct an integrable symplectic map $S_{\gamma}$
in tilde direction, develop an algebro-difference analogue of the Burchnall-Chaundy's theory on
commuting differential operators  by which we express the potential functions  in terms of theta function.
This allows us to derive finite genus solutions  for the lpKP equation in Section 5
and for other two semi-discrete and one continuous pKP equations in Section \ref{sec-6}.
Finally,  concluding remarks are given in the last Section.

\section{ The discretized pKP equations}\label{sec-2}

\subsection{The KP equation}\label{sec-2-1}

In order to find the suitable discrete spectral problems for \eqref{1.1},
let us first recall the usual continuous KP equation,
\begin{equation}
w_t=\frac{1}{4}(w_{xx}+3w^2)_x+\frac{3}{4}\partial^{-1}_xw_{yy}.
\label{1.2}
\end{equation}
It is well-known that the KP equation  has a close relation with the
ZS-AKNS spectral problem $(U_1)$ \cite{6-CaoWG-JMP-1999,16-KonSS-PLA-1991},
\begin{equation}
\partial_x\chi=U_1\chi=\left( \begin{array}{cc}
\lambda/2 &u\\
v &-\lambda/2
\end{array} \right)\chi.
\label{1.4}
\end{equation}
In fact, there is  a hierarchy of isospectral equations $(X_k)$ related to \eqref{1.4},
\begin{equation}
\partial_{\tau_k}(u,v)=X_k,\quad(k=2,3,\cdots),
\label{1.5}
\end{equation}
in which  the first two nonlinear members, ($y=\tau_2,t=\tau_3$),
the NLS equation and the modified KdV (mKdV) equation, respectively,
are
\begin{subequations}\label{1.6}
\begin{align}
&\partial_y(u,v)=X_2=(u_{xx}-2u^2v,\,-v_{xx}+2uv^2),\label{1.6a}\\
&\partial_t(u,v)=X_3=(u_{xxx}-6uvu_x,\,v_{xxx}-6uvv_x).\label{1.6b}
\end{align}
\end{subequations}
Corresponding to the hierarchy \eqref{1.5}, there exist a series of linear spectral problems ($U_k$),
\begin{equation}
\partial_{\tau_k}\chi=U_k\chi,\quad(k=1,2,\cdots),
\label{1.7}
\end{equation}
where, apart from equation (\ref{1.4})$|_{x=\tau_1}$, we also have (with $y=\tau_2,~t=\tau_3$)
\begin{subequations}\label{1.8}
\begin{align}
&\partial_y\chi=U_2\chi=\left( \begin{array}{cc}
\lambda^2/2-uv &\lambda u+u_x\\
\lambda v-v_x &-\lambda^2/2+uv
\end{array} \right)\chi,\label{1.8a}\\
&\partial_t\chi=U_3\chi=\left( \begin{array}{cc}
\lambda^3/2-\lambda uv-u_xv+uv_x &\lambda^2u+\lambda u_x+u_{xx}-2u^2v\\
\lambda^2v-\lambda v_x+v_{xx}-2uv^2 &-\lambda^3/2+\lambda
uv+u_xv-uv_x
\end{array} \right)\chi.\label{1.8b}
\end{align}
\end{subequations}
The Lax pair for $(X_k)$ is composed of $(U_1)$ and $(U_k)$. It is
found that if $(u,v)$ is a compatible solution of $(X_2)$ and
$(X_3)$, then $w=-2uv$ solves the KP equation (\ref{1.2}) \cite{6-CaoWG-JMP-1999,16-KonSS-PLA-1991}.
Thus the compatible conditions of $(U_1)$, $(U_2)$ and $(U_3)$ implies the KP equation. In
other words, $(U_1, U_2, U_3)$ is the Lax triad for the KP equation  and hence for the pKP equation via $w=W_x$.

\subsection{The discrete pKP equations}\label{sec-2-2}

The above facts of the KP equation lead us to consider discretization of the ZS-AKNS problem (\ref{1.4}),
by which we hope to find the Lax representation for the lpKP equation \eqref{1.1}.
One  discretization of \eqref{1.4} is known as the  Ablowitz-Ladik spectral problem \cite{1-AblL-JMP-1975,AblL-JMP-1976},
which leads to a spatially discretized  NLS equation.
In this paper, we  employ the following linear problem, $(D^{(\gamma)})$, adopted in \cite{8-CaoZ-JPA-2012},
\begin{equation}
\tilde\chi=D^{(\gamma)}\chi,\quad
D^{(\gamma)}(\lambda,a,b)=\left(\begin{array}{cc}\lambda-\gamma+ab
&a\\b&1
\end{array}\right),
\label{1.9}
\end{equation}
which provides a second discretzation for (\ref{1.4}) \cite{8a-MerRT-IP-1994} but is different from  Ablowitz-Ladik's  spectral problem (cf.\cite{8b-CheDZ-JNMP-2017}).
Note that \eqref{1.9} is also known as a Darboux transformation of the  ZS-AKNS spectral problems (\ref{1.7})
\cite{AdlY-JPA-1994}.
Here, for the above notation, let $T_1$  be shift operator along the $m_1$  direction, defined for
any function $f:\mathbb{Z}^3\rightarrow\mathbb{R}$ as
\begin{equation*}
(T_1f)(m_1,m_2,m_3)=\tilde{f}(m_1,m_2,m_3)=f(m_1+1,m_2,m_3).
\end{equation*}
Similarly, $T_2f=\bar{f}$, $T_3f=\hat{f}$ are shifts along the $m_2$ and $m_3$ direction, respectively.

Two basic relations,
\begin{equation}
(a,b)=(u,\tilde{v}),\quad(u\tilde{v})_x=\tilde{u}\tilde{v}-uv,
\label{2.1}
\end{equation}
are derived from the compatibility condition of equations (\ref{1.4}) and (\ref{1.9}) (see \cite{8-CaoZ-JPA-2012}).
The former bridges their potential functions, while the latter suggests the setting of difference relation
$\tilde{W}-W=-2u\tilde{v}$ as $W_x=w=-2uv$.
These facts lead to the consideration of three replicas of equation (\ref{1.9}) with distinct
non-zero parameters $\gamma=\gamma_1,\,\gamma_2,\,\gamma_3$,
\begin{subequations}\label{2.2}
\begin{align}
&T_1\chi\equiv\tilde\chi=D^{(\gamma_1)}(\lambda,u,\tilde{v})\chi,\label{2.2a}\\
&T_2\chi\equiv\bar\chi=D^{(\gamma_2)}(\lambda,u,\bar{v})\chi,\label{2.2b}\\
&T_3\chi\equiv\hat\chi=D^{(\gamma_3)}(\lambda,u,\hat{v})\chi,\label{2.2c}
\end{align}
\end{subequations}
which are denoted by $(D^{(\gamma_k)})$, $k=1,2,3$, respectively, for short.
Besides, auxiliary equations will be assigned  for each special occasion from the following list,
\begin{subequations}\label{2.3}
\begin{align}
&\tilde{W}-W=-2u\tilde{v},\label{2.3a}\\
&\bar{W}-W=-2u\bar{v},\label{2.3b}\\
&\hat{W}-W=-2u\hat{v},\label{2.3c}\\
&\partial_xW=-2uv.\label{2.3d}
\end{align}\end{subequations}
At first glance, these equations seem fairly hard to deal with.
Here we remark that on the platform of Liouville integrability, a pair of functions
$(u,v)$ of discrete arguments $m_1,\,m_2,\,m_3$ can be constructed,
which are finite genus potential for each of the discrete spectral
problems (\ref{2.2}a,b,c); and  $W$ can be solved with the help
of the theta function and meromorphic differentials on the associated Riemann surface.
This will lead to  explicit analytic solutions to the lpKP equation (see Section \ref{sec-5}).
The approach can also be extended to the semi-discrete and purely continuous cases (see Section \ref{sec-6}).

To derive the discrete pKP equations \eqref{1.11}, \eqref{1.12} and \eqref{1.13},
we replace  $(U_j)$ in the Lax triad $(U_1,U_2,U_3)$ successively by
$(D^{(\gamma_k)})$, and then we come to the following new Lax triads,
\begin{equation}
(D^{(\gamma_1)},D^{(\gamma_2)},D^{(\gamma_3)}),\,
(U_1,D^{(\gamma_1)},D^{(\gamma_2)}),\,
(U_1,U_2,D^{(\gamma_1)}).
\label{1.10}
\end{equation}
With the auxiliary relations (\ref{2.2}), the compatibility of these triads lead to the discrete pKP
equations \eqref{1.11}, \eqref{1.12} and \eqref{1.13}.
We present the procedure of derivation via the following lemmas and propositions.

\begin{Lemma}\label{L-2.1}
Let $(u,v):\mathbb{Z}^2\rightarrow\mathbb{R}^2$ be a pair of functions such that
$(i)$ equations (\ref{2.2}a,b) have compatible solution $\chi$ for one value of the spectral parameter $\lambda$;
$(ii)$ the system (\ref{2.3}a,b) has a solution  $W$.
Then $(u,v)$ solve  the lNLS equation \cite{8-CaoZ-JPA-2012,15}
\begin{subequations}\label{2.4}
\begin{align}
&\Xi_1^{(0,2)}\equiv(\tilde{u}-\bar{u})(u\bar{\tilde{v}}+1)+(\gamma_1-\gamma_2)u=0,\label{2.4a}\\
&\Xi_2^{(0,2)}\equiv(\tilde{v}-\bar{v})(u\bar{\tilde{v}}+1)-(\gamma_1-\gamma_2)\bar{\tilde{v}}=0,\label{2.4b}
\end{align}
\end{subequations}
and $W,\,u,\,v\,$ satisfy the relation
\begin{align}
Y(\gamma_1,\gamma_2)&\equiv
2 \Bigl(\tilde{v}\,\Xi_1^{(0,2)}(\gamma_1,\gamma_2)+\bar{u}\,\Xi_2^{(0,2)}(\gamma_1,\gamma_2) \Bigr)\nonumber \\
&=\Bigl[\frac{1}{2}(\tilde{W}-\bar{W})+\gamma_1-\gamma_2\Bigr]
(\bar{\tilde{W}}-\tilde{W}-\bar{W}+W)+2(\tilde{u}\tilde{v}-\bar{u}\bar{v}),\label{2.5}
\end{align}
which is equal to zero due to equations (\ref{2.4}).
\end{Lemma}

\begin{proof}
By direct calculations we see that the
cross action $(T_1T_2-T_2T_1)\chi$ is equal to
\begin{equation}
(\tilde{D}^{(\gamma_2)}D^{(\gamma_1)}-\bar{D}^{(\gamma_1)}D^{(\gamma_2)})\chi=
\left(\begin{array}{cc} \Upsilon_{11}&\Xi^{(0,2)}_1\\\Xi^{(0,2)}_2&0
\end{array}\right)
{\chi^{(1)} \choose \chi^{(2)}},
\label{2.6}
\end{equation}
where
\begin{align*}
\Upsilon_{11}=&\frac{\lambda}{\gamma_1-\gamma_2}[(\tilde{v}-\bar{v})\Xi^{(0,2)}_1-(\tilde{u}-\bar{u})\Xi^{(0,2)}_2]\\
&+\frac{1}{\gamma_1-\gamma_2}[(\gamma_1\bar{v}-\gamma_2\tilde{v})\Xi^{(0,2)}_1
-(\gamma_2\bar{u}-\gamma_1\tilde{u})\Xi^{(0,2)}_2].
\end{align*}
With (\ref{2.2}a,b) one can rewrite $(T_1T_2-T_2T_1)\chi$ in the form
\begin{equation}
\frac{1}{\gamma_1-\gamma_2}\left(\begin{array}{cc}
(\gamma_1-\lambda)\bar{\chi}^{(2)}-(\gamma_2-\lambda)\tilde{\chi}^{(2)}
&[(\gamma_1-\lambda)\tilde{u}-(\gamma_2-\lambda)\bar{u}]\chi^{(1)}\\
0&(\gamma_1-\gamma_2)\chi^{(1)}
\end{array}\right){\Xi^{(0,2)}_1 \choose \Xi^{(0,2)}_2 }.
\label{2.7}
\end{equation}
Usually the coefficient determinant is not zero. Since
$(T_1T_2-T_2T_1)\chi=0$, we have $\Xi^{(0,i)}=0$.
Further, in light of equations (\ref{2.3}a,b), the left-hand side of equation (\ref{2.4}) can be written as
\begin{align*}
&\Xi^{(0,2)}_1=(\tilde{u}-\bar{u})+[\frac{1}{2}(\tilde{W}-\bar{W})+\gamma_1-\gamma_2]u,\\
&\Xi^{(0,2)}_2=(\tilde{v}-\bar{v})-[\frac{1}{2}(\tilde{W}-\bar{W})+\gamma_1-\gamma_2]\bar{\tilde{v}},
\end{align*}
which imply equation (\ref{2.5}) by direct calculations.
\end{proof}

\begin{Proposition}\label{P-2.2}
Let
$(u,v):\mathbb{Z}^3\rightarrow\mathbb{R}^2$  be a pair of functions such that
$(i)$ equations (\ref{2.2}a,b,c) have compatible solution $\chi$  for one value of $\lambda$;
$(ii)$ the system (\ref{2.3}a,b,c) has a solution $W$.
Then $W$ solves equation (\ref{1.11}), i.e. $\Xi^{(0,3)}=0$.
\end{Proposition}

\begin{proof}
Consider three replicas of  (\ref{2.5})
with parameters $(\gamma_1,\gamma_2),\,(\gamma_2,\gamma_3),\,(\gamma_3,\gamma_1)$, respectively.
Adding them together we have
\begin{equation}
\Xi^{(0,3)}=Y(\gamma_1,\gamma_2)+Y(\gamma_2,\gamma_3)+Y(\gamma_3,\gamma_1),
\end{equation}
where the terms containing  $u,v$ are canceled. This yields equation
(\ref{1.11}).
\end{proof}

\begin{Proposition}\label{P-2.3}
Let
$(u,v):\mathbb{R}\times\mathbb{Z}^2\rightarrow\mathbb{R}^2$  be a pair of functions such that
$(i)$ equations (\ref{2.2}a,b) have compatible solution $\chi$ for one value of $\lambda$;
$(ii)$ the system of equations (\ref{2.3}a,b) and (\ref{2.3}d) has a solution  $W$.
Then $W$ solves equation (\ref{1.12}), i.e. $\Xi^{(1,2)}=0$.
\end{Proposition}

\begin{proof}
In light of (\ref{2.3}d), the last term in (\ref{2.5}) is equal to $-(\tilde{W}-\bar{W})_x$.
Thus the proof is completed since  $\Xi^{(1,2)}=-Y(\gamma_1,\gamma_2)$.
\end{proof}

\begin{Proposition}\label{P-2.4}
Let
$(u,v):\mathbb{R}^2\times\mathbb{Z}\rightarrow\mathbb{R}^2$  be a pair of
functions such that $(i)$ equations (\ref{1.4}), (\ref{1.8}a) and (\ref{2.2}a) have
compatible solution $\chi$  for one value of  $\lambda$;
$(ii)$ the system of equations (\ref{2.3}a,d) has a solution  $W$.
Then $W$ solves equation (\ref{1.13}), i.e. $\Xi^{(2,1)}=0$.
\end{Proposition}

\begin{proof}
The compatibility condition
$\partial_y\partial_x\chi=\partial_x\partial_y\chi$ gives rise to the NLS
equations (\ref{1.6a}), rewritten as
\begin{subequations}\label{2.9}
\begin{align}
&\Xi_1^{(2,0)}\equiv u_y-u_{xx}+2u^2v=0,\label{2.9a}\\
&\Xi_2^{(2,0)}\equiv v_y+v_{xx}-2uv^2=0.\label{2.9b}
\end{align}
\end{subequations}
In fact,
\begin{align}
(\partial_y\partial_x-\partial_x\partial_y)\chi & =(U_{1,y}-U_{2,x}+[U_1,U_2])\chi\nonumber \\
&=\left(\begin{array}{cc} 0&\Xi^{(2.0)}_1\\
\Xi^{(2.0)}_2&0
\end{array}\right){\chi^{(1)} \choose
\chi^{(2)}}=\left(\begin{array}{cc}\chi^{(2)}&0\\0&\chi^{(1)}
\end{array}\right){\Xi^{(2,0)}_1 \choose \Xi^{(2,0)}_2}.\label{2.10}
\end{align}
Further, the compatibility condition
$\partial_xT_1\chi=T_1\partial_x\chi$ yields the semi-discrete NLS equations
\begin{subequations}\label{2.11}
\begin{align}
&\Xi_1^{(1,1)}\equiv u_x-\tilde{u}-\gamma_1u+u^2\tilde{v}=0,\label{2.11a}\\
&\Xi_2^{(1,1)}\equiv \tilde{v}_x+v+\gamma_1\tilde{v}-u\tilde{v}^2=0,\label{2.11b}
\end{align}
\end{subequations}
since the cross action leads to
\begin{align}
(\partial_xT_1-T_1\partial_x)\chi &=(D^{(\gamma_1)}_x-\tilde{U}_1D^{(\gamma_1)}+D^{(\gamma_1)}U_1)\chi \nonumber\\
&=\left(\begin{array}{cc} \kappa_{11}&\Xi^{(1,1)}_1\\\Xi^{(1,1)}_2&0
\end{array}\right){\chi^{(1)} \choose
\chi^{(2)}}=\left(\begin{array}{cc}\tilde{\chi}^{(2)}&u\chi^{(1)}\\0&\chi^{(1)}
\end{array}\right){\Xi^{(1,1)}_1 \choose \Xi^{(1,1)}_2},
\label{2.12}
\end{align}
where
$\kappa_{11}=(u\tilde{v})_x-\tilde{u}\tilde{v}+uv=\tilde{v}\Xi^{(1,1)}_1+u\Xi^{(1,1)}_2$,
and the relation $\tilde{\chi}^{(2)}=\tilde{v}\chi^{(1)}+\chi^{(2)}$ has been used.
By calculation we have
\begin{align*}
&\tilde{v}\Xi^{(2,0)}_1+u\tilde{\Xi}^{(2,0)}_2
=(u\tilde{v})_y-(u_x\tilde{v}-u\tilde{v}_x)_x-2u\tilde{v}(\tilde{u}\tilde{v}-uv),\\
&(\tilde{v}\Xi^{(1,1)}_1-u\Xi^{(1,1)}_2)_x
=(u_x\tilde{v}-u\tilde{v}_x)_x-(\tilde{u}\tilde{v}+uv+2\gamma_1u\tilde{v}-2u^2\tilde{v}^2)_x.
\end{align*}
Adding them together  we arrive at
\begin{equation}
\Xi^{(2,1)}=-2(\tilde{v}\Xi^{(2,0)}_1+u\tilde{\Xi}^{(2,0)}_2)-2(\tilde{v}\Xi^{(1,1)}_1-u\Xi^{(1,1)}_2)_x,
\label{2.13}
\end{equation}
where the term $(u_x\tilde{v}-u\tilde{v}_x)_x$ is canceled and the
variable $W$ is introduced by equations (\ref{2.3}a,d). Thus
$\Xi^{(2,1)}=0$.
\end{proof}

Let us  back to the equations \eqref{1.11}, \eqref{1.12} and \eqref{1.13}.
We have seen that (\ref{1.11}) is nothing but  the lpKP equation \eqref{1.1}.
Besides, equations (\ref{1.12}) and (\ref{1.13}) have close relations with the (N-2) and (N-3)
models that were discovered by Date, Jimbo and Miwa\cite{9-DatJM-JPSJ-1982}, which are
\begin{align}
&\Xi^{\mathrm{N}2}\equiv(\tilde{V}-\bar{V})_x-(e^{\bar{\tilde{V}}}-e^{\tilde{V}}-e^{\bar{V}}+e^{V})=0,\label{1.15}\\
&\Xi^{\mathrm{N}3}\equiv\Delta(V_y+\frac{2}{h}V_x-2VV_x)-(\Delta+2)V_{xx}=0,\label{1.16}
\end{align}
where $\Delta f=\tilde{f}-f$ for arbitrary function $f$.
In fact, for  (\ref{1.12}), introducing
\begin{equation}
V=\ln[(\tilde{W}-\bar{W})/2+\gamma_1-\gamma_2],
\label{1.17}
\end{equation}
and then using  (\ref{1.12}), one finds
\begin{equation*}
V_x=\frac{(\tilde{W}-\bar{W})_x/2}{(\tilde{W}-\bar{W})/2+\gamma_1-\gamma_2}
=\frac{1}{2}(\bar{\tilde{W}}-\tilde{W}-\bar{W}+W).
\end{equation*}
It then follows that $(\tilde{V}-\bar{V})_x$ is equal to the second part in equation
(\ref{1.15}). Hence $\Xi^{\text{N2}}=0$.
Thus, for any solution $W$ of the equation \eqref{1.12}, $V$ defined by \eqref{1.17}
provides a special solution for \eqref{1.15}.
For the equation (\ref{1.13}), if  $W$ is a solution, then
\begin{equation}
V=(\tilde{W}-W)/2
\end{equation}
solves the (N-3) equation (\ref{1.16}). Actually, it is easy to find (with $\gamma_1=1/h$)
\begin{equation*}
\frac{1}{2}\Xi^{(2,1)}=(V_y+\frac{2}{h}V_x-2VV_x)-(V+W)_{xx},
\end{equation*}
which implies $\Xi^{\mathrm{N3}}=\Delta\Xi^{(2,1)}/2=0$.
In this sense, equation (\ref{1.13}) is the potential version of (N-3).
Note that  (N-3) model was also derived by Kanaga Vel and Tamizhmani, with the help of
quasi-difference operators \cite{14-KanT-CSF-1997}, known as the D$\Delta$KP equation
there. Besides, some properties of the D$\Delta$KP hierarchy, including symmetries, Hamiltonian structures and  continuum limit,
were investigated in \cite{11-FuHTZ-Non-2013}.

At the end of this subsection, we consider continuum limits of equations (\ref{1.11},\ref{1.12},\ref{1.13}).
Let $\gamma_k=-1/\varepsilon_k$,  $\varepsilon_k=c_k\varepsilon$,
$(k=1,2,3)$, where $c_1,\,c_2,\,c_3$ are arbitrary distinct non-zero constants.
For any smooth function  $W(x,y,t)$, define
\begin{equation}\label{2.14}
T_kW=W(x+c_k\varepsilon,\,y-c_k^2\varepsilon^2/2,\,t+c_k^3\varepsilon^3/3),\quad
(k=1,2,3).
\end{equation}
Denote $T_1 W=\tilde{W},\,T_2W=\bar{W},\,T_3W=\hat{W}$  for short.
By straightforward calculations we have the following.
\begin{Proposition}\label{P-2.5}
Under the Ansatz (\ref{2.14}), in the
neighborhood of $\varepsilon\sim 0$, the following Taylor expansions
hold for any smooth function  $W(x,y,t)$,
\begin{subequations}\label{2.15}
\begin{align}
&\Xi^{(2,1)}=\Xi^{(3,0)}\frac{2c^2_1}{3}\varepsilon^2+O(\varepsilon^3), \label{2.15a}\\
&\Xi^{(1,2)}=\Xi^{(3,0)}\frac{c_1c_2(c_1-c_2)}{3}\varepsilon^3+O(\varepsilon^4), \label{2.15b}\\
&\Xi^{(0,3)}=\Xi^{(3,0)}\frac{1}{3}[c_1c_2(c_2-c_1)+c_2c_3(c_3-c_2)+c_3c_1(c_1-c_3)]\varepsilon^3+O(\varepsilon^4). \label{2.15c}
\end{align}
\end{subequations}
\end{Proposition}
Thus, all the continuum limits of the lattice pKP equations (\ref{1.11}), (\ref{1.12}) and (\ref{1.13})
give rise to the same pKP equation (\ref{1.3}).
The Ansatz (\ref{2.14}) is crucial, which is proposed based on
comparing the velocities of the Abel-Jacobi variable $\vec{\phi}$
along the discrete $S_{\gamma_k}$-flow and the continuous $H_j$-flow  (see Appendix \ref{A-1}).

\section{ The integrable Hamiltonian system $(H_1)$}\label{sec-3}

In \cite{6-CaoWG-JMP-1999,8-CaoZ-JPA-2012}  an integrable Hamiltonian system
is constructed from the ZS-AKNS spectral problem,
\begin{subequations}\label{3.1}
\begin{align}
&\partial_x{p_j \choose q_j}={-\partial H_1/\partial q_j \choose
\partial H_1/\partial p_j}= \left(\begin{array}{cc}
\alpha_j/2&-<p,p>\\<q,q>&-\alpha_j/2
\end{array}\right){p_j \choose q_j},\label{3.1a}\\
&H_1(p,q)=-\frac{1}{2}<Ap,q>+\frac{1}{2}<p,p><q,q>.\label{3.1b}
\end{align}
\end{subequations}
where $A=\mathrm{diag}(\alpha_1,\cdots,\alpha_N)$,
$<\xi,\eta>=\sum^N_{j=1}\xi_j\eta_j$. It can be regarded as $N$
replicas of equation (\ref{1.4}) with eigenvalues
$\alpha_1,\cdots,\alpha_N$, respectively, under the constraint
\begin{equation}
(u,v)=f_{U_1}(p,q)=(-<p,p>,<q,q>).
\label{3.2}
\end{equation}
The integrability requires enough number of involutive integrals. In deriving them, we use the Lax equation
\begin{equation}
\partial_xL(\lambda)=[U_1(\lambda),L(\lambda)],
\label{3.3}
\end{equation}
which has a solution, the Lax matrix  \cite{6-CaoWG-JMP-1999,8-CaoZ-JPA-2012}
\begin{equation}
L(\lambda;p,q)=\left(\begin{array}{cc}
1/2+Q_{\lambda}(p,q)&-Q_{\lambda}(p,p)\\Q_{\lambda}(q,q)&-1/2-Q_{\lambda}(p,q)
\end{array} \right),
\label{3.4}
\end{equation}
where $Q_{\lambda}(\xi,\eta)=<(\lambda I-A)^{-1}\xi,\eta>$. By
equation (\ref{3.3}), $F(\lambda)=\text{det}L(\lambda)$  is independent of
the argument $x$. Three sets of integrals are derived from the expansions
\begin{subequations}\label{3.5}
\begin{align}
&F(\lambda)=-\frac{1}{4}+\sum_{k=1}^N\frac{E_k}{\lambda-\alpha_k}
=-\frac{1}{4}+\sum_{j=0}^{\infty}F_j\lambda^{-j-1},\label{3.5a} \\
&H(\lambda)=\sqrt{-F(\lambda)}=\frac{1}{2}-2\sum_{k=0}^{\infty}H_k\lambda^{-k-1}, \label{3.5b}
\end{align}
\end{subequations}
with $F_0=-<p,q>$, $H_0=-<p,q>/2$, $H_1$ exactly the same as in
equation (\ref{3.1b}), and
\begin{subequations}\label{3.6}
\begin{align}
&E_k=-p_kq_k+\sum_{1\leq j\leq N;\,j\neq k}\frac{(p_jq_k-p_kq_j)^2}{\alpha_k-\alpha_j},\label{3.6a}\\
&F_k=-<A^kp,q>+\sum_{\begin{subarray}{c}i+j=k-1;\\
i,\,j\geq 0\end{subarray}}(<A^ip,p><A^jq,q>-<A^ip,q><A^jp,q>),\label{3.6b}\\
&H_k=\frac{1}{2}F_k+2\sum_{\begin{subarray}{c}i+j=k-1;\\
i,\,j\geq 0\end{subarray}}H_iH_j.\label{3.6c}
\end{align}
\end{subequations}
The functions $\{E_k\}$ are called confocal polynomials, satisfying
\begin{equation}\label{3.7}
\sum_{k=1}^N\alpha_k^jE_k=F_j,\quad\sum_{k=1}^NE_k=F_0=-<p,q>.
\end{equation}
Further, we have the Lax equation along the $F(\lambda)$-flow,
\begin{equation}\label{3.8}
\frac{\text{d}}{\text{d}t_{\lambda}}L(\mu)= \{L(\mu),F(\lambda)\}=\frac{2}{\lambda-\mu}[L(\lambda),L(\mu)],
\end{equation}
which can be verified directly. It implies $\{F(\mu),F(\lambda)\}=\partial_{t_{\lambda}}\text{det}L(\mu)=0$.
Here $\{A,B\}$ is the usual Poisson bracket defined as
\[\{A,B\}=\sum^{N}_{k=1}\biggl(\frac{\partial A}{\partial q_k}\frac{\partial B}{\partial p_k}-
\frac{\partial A}{\partial p_k}\frac{\partial B}{\partial q_k}\biggr).
\]
As a corollary we have

\begin{Lemma}\label{L-3.1}
The members in the set $\{E_j,F_k,H_l\}$ are involutive in pairs.
\end{Lemma}

By \cite{6-CaoWG-JMP-1999}, there is an inner relation between the integral $H_k$ and
$(X_k)$, the AKNS equation (\ref{1.5}). The involutivity $\{H_1,H_k\}=0$
implies the commutativity of the Hamiltonian phase flows
$g_{H_1}^x,\,g_{H_k}^{\tau_k}$. This yields a compatible solution
for $(H_1),\,(H_k)$, and hence a solution to equation $(X_k)$,
respectively, as
\begin{subequations}\label{g-flow}
\begin{align}
&(p(x,\tau_k),\,q(x,\tau_k))=g_{H_1}^x g_{H_k}^{\tau_k} (p_0,q_0), \\
&(u(x,\tau_k),\,v(x,\tau_k))=f_{U_1}(p,q)=(-<p,p>,<q,q>).
\end{align}
\end{subequations}

Let $\alpha(\lambda)=\Pi_{k=1}^N(\lambda-\alpha_k)$. By \cite{6-CaoWG-JMP-1999,8-CaoZ-JPA-2012}, a
curve $\mathcal{R}:\xi^2=R(\lambda)$, with genus $g=N-1$, is
constructed by the factorization of
$F(\lambda)=-\Lambda(\lambda)/[4\alpha(\lambda)]$, with
$R(\lambda)=\Lambda(\lambda)\alpha({\lambda})$. For non-branching
$\lambda$, there are two points $\mathfrak{p}(\lambda)$,
$\tau\mathfrak{p}(\lambda)$ on $\mathcal{R}$, with
$\tau: \mathcal{R}\rightarrow\mathcal{R}$ the map of changing sheets.
Consider two objects on the curve, the canonical basis
$a_1,\cdots,a_g,b_1,\cdots,b_g$ of homology group of contours, and
the basis of holomorphic differentials, written in the vector form
as $\vec{\omega}^\prime=(\omega_1^\prime,\cdots,\omega_g^\prime)^T$,
$\omega^{\prime}_j=\lambda^{g-j}\text{d}\lambda/(2\xi)$. It is
normalized into $\vec{\omega}=C\vec{\omega}^\prime$, where
$C=(a_{jk})^{-1}_{g\times g}$, with $a_{jk}$ the integral of
$\omega_j^\prime$ along $a_k$. Near the infinities, the local
expansions have simple relation as
\begin{equation}\label{3.9}
\vec\omega=\begin{cases} +(\vec\Omega_1+\vec\Omega_2 z+\vec\Omega_3
z^2+\cdots)\text{d}z,&\text{near $\infty_+$,}\\
-(\vec\Omega_1+\vec\Omega_2 z+\vec\Omega_3
z^2+\cdots)\text{d}z,&\text{near $\infty_-$.}
\end{cases}
\end{equation}
Periodic vectors  $\vec\delta_k$ and $\vec B_k$ are defined as integrals
of $\vec\omega$ along $a_k$ and $b_k$, respectively. They span a lattice
$\mathscr T$, which defines the Jacobian variety $J(\mathcal
R)=\mathbb C^g/\mathscr T$. The Abel map $\mathscr A(\mathfrak p)$
is given as the integral of $\vec\omega$  from the fixed point
$\mathfrak{p}_0$ to  $\mathfrak{p}$. The matrix  $B$, with $\vec
B_k$  as columns, is used to construct the theta function
$\theta(\vec z,B)$.

The elliptic variables $\mu_j,\,\nu_j$ are given by the roots of the
off-diagonal entries of the Lax matrix,
\begin{subequations}\label{3.10}
\begin{align}
&L^{12}(\lambda)=-<p,p>\frac{\mathfrak
m(\lambda)}{\alpha(\lambda)},\quad \mathfrak
m(\lambda)=\Pi_{j=1}^g(\lambda-\mu_j), \label{3.10a}\\
& L^{21}(\lambda)=<q,q>\frac{\mathfrak
n(\lambda)}{\alpha(\lambda)},\quad \mathfrak
n(\lambda)=\Pi_{j=1}^g(\lambda-\nu_j). \label{3.10b}
\end{align}
\end{subequations}
They define the quasi-Abel-Jacobi and Abel-Jacobi variables,
respectively, as
\begin{subequations}
\label{3.11}
\begin{align}
&\vec\psi^\prime=\sum_{k=1}^g\int_{\mathfrak p_0}^{\mathfrak p(\mu_k)}\vec\omega^\prime,
\quad\vec\psi=C\vec\psi^\prime=\mathscr {A}(\sum_{k=1}^g\mathfrak p(\mu_k)),\label{3.11a}\\
&\vec\phi^\prime=\sum_{k=1}^g\int_{\mathfrak p_0}^{\mathfrak
p(\nu_k)}\vec\omega^\prime,\quad\vec\phi=C\vec\phi^\prime=\mathscr
{A}(\sum_{k=1}^g\mathfrak p(\nu_k)).\label{3.11b}
\end{align}
\end{subequations}
The evolution of these two variables along the $F(\lambda)$-flow is
obtained by equation (\ref{3.8}). Actually, in the component equation for
$L^{21}(\mu)$, by letting $\mu\rightarrow\nu_k$, we calculate
\begin{subequations}\label{3.12}
\begin{align}
&\frac{1}{2\sqrt {R(\nu_k)}}\frac{\text{d}\nu_k}{\text{d}t_\lambda}
=\frac{-\mathfrak{n}(\lambda)}{\alpha(\lambda)(\lambda-\nu_k)\mathfrak{n}^\prime(\nu_k)},\label{3.12a}\\
&\{\phi_s^{\prime},F(\lambda)\}=\frac{\text{d}\phi_s^{\prime}}{\text{d}t_\lambda}
=\sum_{k=1}^g\frac{\nu_k^{g-s}}{2\sqrt
{R(\nu_k)}}\frac{\text{d}\nu_k}{\text{d}t_\lambda}=-\frac{\lambda^{g-s}}{\alpha(\lambda)},\label{3.12b}
\end{align}
\end{subequations}
where $\vec\phi^{\prime}=(\phi_1^\prime,\cdots,\phi_g^\prime)$.
By the partial fraction expansion (\ref{3.5a}), we get
\begin{subequations}\label{3.13}
\begin{align}
&\{\phi_s^{\prime},E_k\}=-\alpha_k^{g-s}/\alpha^\prime(\alpha_k),\quad(k=1,\cdots,N),\label{3.13a}\\
&\{\phi_s^{\prime},E_1+\cdots+E_N\}=\{\phi_s^{\prime},F_0\}=0.\label{3.13b}
\end{align}
\end{subequations}

\begin{Proposition}\label{P-3.2}
Each Hamiltonian system $(H_k)$,
$k=1,2,\cdots$, is integrable in Liouville sense, sharing the same
integrals $E_1,\cdots, E_N$, which are involutive in pairs and
functionally independent in  $\mathbb R^{2N}-\{0\}$.
\end{Proposition}

\begin{proof} According to Lemma \ref{L-3.1}, it only needs to prove
the functional independence of the confocal polynomials. Suppose
$\sum_{k=1}^Nc_k\text{d}E_k=0$. Then
$\sum_{k=1}^Nc_k\{\phi_s^\prime,E_k\}=0$. By equation (\ref{3.13b}), we
have
\begin{equation}
\sum_{k=1}^g(c_k-c_N)\{\phi_s^\prime,E_k\}=0,\quad(1\leq s\leq
g).\notag
\end{equation}
The coefficient matrix is non-degenerate since by equation (\ref{3.13}a)
it is of Vandermonde type. Hence we have $c_k-c_N=0$ and
$c_N\sum_{k=1}^N\text{d}E_k=0$. This implies $c_N=0$ since
\begin{equation*}
\sum_{k=1}^N\text{d}E_k=-\text{d}<p,q>=-\sum_{j=1}^N(q_j\text{d}p_j+p_j\text{d}q_j)\neq 0.
\end{equation*}
\end{proof}

\begin{Lemma}\label{L-3.3}
The Abel-Jacobi variables straightens
out the  $H(\lambda)$-flow as
\begin{equation}\label{3.14}
\{\vec\phi,H(\lambda)\text{d}\lambda\}=2\vec\omega.
\end{equation}
\end{Lemma}
\begin{proof}
Since $F(\lambda)=-H^2(\lambda)$, equation
(\ref{3.12b}) is transformed into the following formula,
which,  by multiplied the matrix  $C$, leads to
equation (\ref{3.14}),
\begin{equation*}
\{\phi_s^\prime,H(\lambda)\text{d}\lambda\}=\frac{\lambda^{g-s}\text{d}\lambda}{2H(\lambda)\alpha(\lambda)}
=\frac{\lambda^{g-s}\text{d}\lambda}{\sqrt{R(\lambda)}}=2\omega_s^\prime.
\end{equation*}
\end{proof}

\begin{Proposition}\label{P-3.4}
The Abel-Jacobi variables straighten out the $H_k$-flow, as $\{\vec\phi,H_0\}=0$  and
\begin{subequations}\label{3.15}
\begin{align}
&\frac{\text{d}\vec\phi}{\text{d}\tau_k}=\{\vec\phi,H_k\}=\vec\Omega_k,\quad(k=1,2,\cdots), \label{3.15a}\\
&\vec\phi(\tau_k)\equiv\vec\phi(0)+\tau_k\vec\Omega_k,\quad(\mathrm{mod}\,\mathscr{T}). \label{3.15b}
\end{align}
\end{subequations}
\end{Proposition}
\begin{proof}
By the equations (\ref{3.5}b) and (\ref{3.9}), we have an expansion of equation (\ref{3.14}) near $\infty_{+}$.
Equation (\ref{3.15a}) is then obtained as its coefficient.
\end{proof}

For another Abel-Jacobi variable, by the equation (3.25) in \cite{8-CaoZ-JPA-2012}, we have
\begin{subequations}\label{3.16}
\begin{align}
&\vec\psi+\vec\eta_+ \equiv \vec\phi+\vec\eta_-,\quad(\text{mod}\,\mathscr{T}),\label{3.16a}\\
&\vec\eta_{\pm}=\int_{\infty_{\pm}}^{\mathfrak{p}_0}\vec\omega,\quad\vec\Omega_D
=\vec\eta_+-\vec\eta_-=\int_{\infty_+}^{\infty_-}\vec\omega,\label{3.16b}\\
&\vec\psi(\tau_k) \equiv \vec\phi(\tau_k)-\vec\Omega_D \equiv
\vec\psi(0)+\tau_k\vec\Omega_k,\quad(\text{mod}\,\mathscr{T}).\label{3.16c}
\end{align}
\end{subequations}

\section{ The integrable map $S_\gamma$}\label{sec-4}

In \cite{8-CaoZ-JPA-2012}, an integrable symplectic map
$S_{\gamma}:\mathbb{R}^{2N}\rightarrow\mathbb{R}^{2N}$,
$(p,q)\mapsto(\tilde{p},\tilde{q})$, is constructed with the help of
$N$ replicas of discrete ZS-AKNS equation (\ref{1.9}),
\begin{equation}\label{4.1}
{\tilde{p}_j \choose
\tilde{q}_j}=(\alpha_j-\gamma)^{-1/2}D^{(\gamma)}(\alpha_j;a,b){p_j\choose
q_j},\quad(1\leq j\leq N),
\end{equation}
under the discrete constraint  $(a,b)=f_\gamma(p,q)$,
\begin{equation}\label{4.2}
a=-<p,p>,\quad
b=\frac{1}{Q_\gamma(p,p)}\Big(-\frac{1}{2}-Q_\gamma(p,q)\pm\frac{\sqrt{R(\gamma)}}{2\alpha(\gamma)}
\Big).
\end{equation}
It can be derived from the continuous constraint (\ref{3.2}) through the relation $a=u, b=\tilde{v}$ in (\ref{2.1}).
In fact,
\begin{align*}
\tilde{v}-b&=<\tilde{q},\tilde{q}>-b=<(A-\gamma
I)^{-1}(bp+q),bp+q>-b\\
&=b^2L^{12}(\gamma)-2bL^{11}(\gamma)-L^{21}(\gamma)\equiv
P^{(\gamma)}(b).
\end{align*}
Thus $P^{(\gamma)}(b)=0$, whose roots lead to equation (\ref{4.2}). The
factor $(\alpha_j-\gamma)^{-\frac{1}{2}}$ in equation (\ref{4.1}) is introduced
so that the coefficient determinant equals to unity,
which is necessary for making the resulting map $S_\gamma$ symplectic.

As in the continuous case, the Liouville integrability of the map $S_\gamma$ requires enough number of involutive integrals.
Similarly, the discrete Lax equation, given as follows, plays a central role,
\begin{equation}\label{4.3}
L(\lambda;\tilde{p},\tilde{q})D^{(\gamma)}(\lambda;a,b)=D^{(\gamma)}(\lambda;a,b)L(\lambda;p,q).
\end{equation}
By \cite{8-CaoZ-JPA-2012}, under the constraint (\ref{4.2}), it has the same Lax matrix,
given by equation (\ref{3.4}), as its solution. Immediately we have
$F(\lambda;\tilde{p},\tilde{q})=F(\lambda;p,q)$ by taking the determinant of (\ref{4.3}).
Thus $F(\lambda)$, together with $H(\lambda)$, $E_j,\,F_k,\,H_l$, are all invariant under the action
of the map $S_{\gamma}$.

\begin{Proposition}\label{P-4.1}
\cite{8-CaoZ-JPA-2012} The map $S_{\gamma}$ is
symplectic and integrable, possessing $F(\lambda)$, $\{F_j\}$,
$\{H_l\}$ and the confocal polynomials $E_1,\cdots,E_N$, as its
integrals.
\end{Proposition}

Construct a discrete flow
\begin{equation}\label{S-flow}
(p(m),q(m))=S_{\gamma}^m(p_0,q_0)
\end{equation}
by iteration.
It generates the finite genus potential functions for equation (\ref{1.9}),
\begin{equation}\label{4.4}
(a_m,b_m)=(u_m,v_{m+1})=(-<p,p>,<\tilde q,\tilde q>).
\end{equation}
Define $L_m(\lambda)=L(\lambda;p(m),q(m))$,
$D_m^{(\gamma)}(\lambda)=D^{(\gamma)}(\lambda;u_m,v_{m+1})$.
Rewrite equation (\ref{4.3}) as
\begin{equation}\label{4.5}
L_{m+1}(\lambda)D_m^{(\gamma)}(\lambda)=D_m^{(\gamma)}(\lambda)L_{m}(\lambda).
\end{equation}
Consider the discrete ZS-AKNS problem (\ref{1.9}) with finite genus potential functions as
\begin{equation}\label{4.6}
h(m+1,\lambda)=D_m^{(\gamma)}(\lambda)h(m,\lambda).
\end{equation}
The solution space $\mathcal{E}_{\lambda}$ is invariant under the
action of $L_m(\lambda)$ due to the commutativity relation (\ref{4.5}).
The linear operator $L_m(\lambda)$ has eigenvalues $\pm H(\lambda)$,
with associated eigenvectors $h_{\pm}$ in $\mathcal{E}_{\lambda}$,
satisfying
\begin{subequations}\label{4.7}
\begin{align}
&L_m(\lambda)h_{\pm}(m,\lambda)=\pm H(\lambda)h_{\pm}(m,\lambda),\label{4.7a}\\
&h_{\pm}(m+1,\lambda)=D_m^{(\gamma)}(\lambda)h_{\pm}(m,\lambda).\label{4.7b}
\end{align}
\end{subequations}

Roughly speaking, the situation can be regarded as an
algebro-difference analogue of the Burchnall-Chaundy's theory on
commuting differential operators \cite{5,5b}. Actually, let
$\mathcal{L}(\lambda)=2\alpha(\lambda)L(\lambda)$.
Then $\det\mathcal{L}(\lambda)=-R(\lambda)$  is a polynomial  rather than a rational function.
The commutativity relation (\ref{4.5}) is rewritten as
$\mathcal{L}_{m+1}D_m^{(\gamma)}=D_m^{(\gamma)}\mathcal{L}_m$. The
algebraic spectral problem (\ref{4.7a}) is revised as $\mathcal{L}_m
h_{\pm}=\xi h_{\pm}$, with $\xi=\pm\sqrt{R(\lambda)}$. The algebraic
problem and the difference problem share common eigenvectors
$h_{\pm}$, with eigenvalues satisfying the algebraic relation,
$\xi^2=R(\lambda)$, exactly the same as the affine equation of the
algebraic curve $\mathcal{R}$.

Let $M(m,\lambda)$  be fundamental solution matrix of equation
(\ref{4.6}). Under the normalization condition
$h_{\pm}^{(2)}(0,\lambda)=1$, the eigenvectors are determined
uniquely as
\begin{subequations}\label{4.8}
\begin{align}
&h_{\pm}(m,\lambda)={h_{\pm}^{(1)}(m,\lambda)\choose
h_{\pm}^{(2)}(m,\lambda)}=M(m,\lambda){c_{\lambda}^{\pm}\choose 1}, \label{4.8a}\\
&c^{\pm}_{\lambda}=\frac{L_0^{11}(\lambda)\pm
H(\lambda)}{L_0^{21}(\lambda)}=\frac{-L_0^{12}(\lambda)}{L_0^{11}(\lambda)\mp
H(\lambda)}. \label{4.8b}
\end{align}
\end{subequations}
Two meromorphic functions, the Baker functions,
$\mathfrak{h}^{(\kappa)}(m,\mathfrak{p})$,
$\mathfrak{p}\in\mathcal{R}$, $\kappa=1,2$, are defined as
\begin{equation}\label{4.9}
\mathfrak{h}^{(\kappa)}(m,\mathfrak{p}(\lambda))=h_+^{(\kappa)}(m,\lambda),\quad
\mathfrak{h}^{(\kappa)}(m,\tau\mathfrak{p}(\lambda))=h_-^{(\kappa)}(m,\lambda).
\end{equation}
The commutativity relation (\ref{4.5}) implies formulas of
Dubrovin-Novikov's type \cite{8-CaoZ-JPA-2012}. They are applied to calculate the
divisors of the Baker functions. This leads to the straightening out
of the flow  $S_{\gamma}^m$ on the Jacobian variety as [8]
\begin{subequations}\label{4.10}
\begin{align}
&\vec{\psi}(m)\equiv\vec{\phi}(0)+m\vec\Omega_{\gamma}-\vec{\Omega}_D,\quad(\text{mod}\,\mathscr{T}),\label{4.10a}\\
&\vec{\phi}(m)\equiv\vec{\phi}(0)+m\vec\Omega_{\gamma},\quad\quad(\text{mod}\,\mathscr{T}),\label{4.10b}\\
&\vec{\Omega}_{\gamma}=\int_{\mathfrak{p}(\gamma)}^{\infty_+}\vec{\omega},\quad\vec{\Omega}_D
=\int_{\infty_+}^{\infty_-}\vec{\omega},\label{4.10c}
\end{align}
\end{subequations}
where the Abel-Jacobi variables are given by   (\ref{3.11}) as
\[
\vec{\psi}(m)=\mathscr{A}\big(\sum_{j=1}^g\mathfrak{p}(\mu_j(m))\big),\quad
\vec{\phi}(m)=\mathscr{A}\big(\sum_{j=1}^g\mathfrak{p}(\nu_j(m))\big).
\]
For any two distinct points
$\mathfrak{q},\mathfrak{r}\in\mathcal{R}$, there exists a dipole
$\omega[\mathfrak{q},\mathfrak{r}]$, an Abel differential of the
third kind, with residues $1$ and $-1$ at the poles
$\mathfrak{q},\mathfrak{r}$, respectively, satisfying \cite{20}
\begin{equation}\label{4.11}
\int_{a_j}\omega[\mathfrak{q},\mathfrak{r}]=0,\quad\int_{b_j}\omega[\mathfrak{q},\mathfrak{r}]
=\int_{\mathfrak{r}}^{\mathfrak{q}}\omega_j,
\quad(j=1,\cdots,g).
\end{equation}
With the help of these dipoles, the Baker functions can be
reconstructed as \cite{8-CaoZ-JPA-2012}
\begin{subequations}\label{4.12}
\begin{align}
&\mathfrak{h}^{(1)}(m,\mathfrak{p})=d_m^{(1)}\frac{\theta[-\mathscr{A}(\mathfrak{p})+\vec{\psi}(m)+\vec{K}]}
{\theta[-\mathscr{A}(\mathfrak{p})+\vec{\phi}(0)+\vec{K}]}
e^{\int_{\mathfrak{p}_0}^{\mathfrak{p}}m\omega[\mathfrak{p}(\gamma),\infty_+]
+\omega[\infty_-,\infty_+]},\label{4.12a}\\
&\mathfrak{h}^{(2)}(m,\mathfrak{p})=d_m^{(2)}\frac{\theta[-\mathscr{A}(\mathfrak{p})+\vec{\phi}(m)+\vec{K}]}
{\theta[-\mathscr{A}(\mathfrak{p})+\vec{\phi}(0)+\vec{K}]}
e^{\int_{\mathfrak{p}_0}^{\mathfrak{p}}m\omega[\mathfrak{p}(\gamma),\infty_+]},\label{4.12b}
\end{align}
\end{subequations}
where $d_m^{(1)},\,d_m^{(2)}$  and $\vec{K}$  are constants,
independent of $\mathfrak{p}\in\mathcal{R}$.

With these results in hand, we start to derive an explicit formula for the function  $u\tilde{v}$.
To this end we consider the local expression of the dipole near  $\infty_+$, ( $z=\lambda^{-1}$ ),
\begin{equation}\label{4.13}
\omega[\mathfrak{p}(\gamma),\infty_+]=[-z^{-1}+\varphi(z)]\text{d}z,
\end{equation}
with $\varphi(z)$ holomorphic near  $z\sim 0$. A simple calculation
yields
\begin{equation}\label{4.14}
\partial_z
\text{log}(z\,\text{exp}\int_{\mathfrak{p}_0}^{\mathfrak{p}}\omega[\mathfrak{p}(\gamma),\infty_+])=\varphi(z).
\end{equation}
Recalling equation (\ref{3.9}), we have
\begin{equation}\label{4.15}
-\mathscr{A}(\mathfrak{p})=\vec{\eta}_+-\vec{\Omega}_1z+O(z^2),\quad\vec{\eta}_+
=\int_{\infty_+}^{\mathfrak{p}_0}\vec{\omega}.
\end{equation}
Then, from   (\ref{4.12a}) we get
\begin{subequations}\label{4.16}
\begin{align}
&\frac{z\tilde{h}_+^{(1)}}{h_+^{(1)}}=\frac{d_{m+1}}{d_m}\frac{\theta[-\vec{\Omega}_1 z+O(z^2)+\vec{\eta}_++\vec{\psi}(m+1)+\vec{K}]}
{\theta[-\vec{\Omega}_1 z+O(z^2)+\vec{\eta}_++\vec{\psi}(m)+\vec{K}]}\cdot z e^{\int_{\mathfrak{p}_0}^{\mathfrak{p}}\omega[\mathfrak{p}(\gamma),\infty_+]},
\label{4.16a}\\
&\partial_z\text{log}\frac{z\tilde{h}_+^{(1)}}{h_+^{(1)}}=\partial_z\text{log}\frac{\theta[-\vec{\Omega}_1
z+O(z^2)+\vec{\eta}_++\vec{\psi}(m+1)+\vec{K}]}
{\theta[-\vec{\Omega}_1
z+O(z^2)+\vec{\eta}_++\vec{\psi}(m)+\vec{K}]}+\varphi(z).\label{4.16b}
\end{align}
\end{subequations}
On the other hand, since $h_{\pm}=(h_{\pm}^{(1)},h_{\pm}^{(2)})^T$
satisfies equation (\ref{4.7b}), we have
\begin{equation}\label{4.17}
\frac{z\tilde{h}_+^{(1)}}{h_+^{(1)}}=1+(u\tilde{v}-\gamma)z+\frac{uh_+^{(2)}}{h_+^{(1)}}z
=1+(u\tilde{v}-\gamma)z+O(z^2),
\end{equation}
where the following estimation is used,
\begin{equation}
\frac{uh_+^{(2)}}{h_+^{(1)}}=\frac{L^{11}(\lambda)-H(\lambda)}{-L^{12}(\lambda)}
=<q,q>\lambda^{-1}[1+O(\lambda^{-1})]=O(z).\notag
\end{equation}
Now, taking derivative of the equation (\ref{4.17}) with respect to $z$ and comparing it
with  (\ref{4.16b}) at  $z=0$, with the help of the relation
$\vec{\psi}+\vec{\eta}_+\equiv\vec{\phi}+\vec{\eta}_-$ in equation
(\ref{3.16a}), we obtain the following.

\begin{Proposition}\label{P-4.2}
Let $(a,b)=(u,\tilde{v})$ be finite genus potential functions of equation (\ref{1.9}), defined by equation (\ref{4.4}).
Then we have
\begin{equation}\label{4.18}
u\tilde{v}=-\partial_z\vert_{z=0}\mathrm{log}\frac{\theta[\vec{\Omega}_1
z+\vec{\phi}(m+1)+\vec{\eta}_-+\vec{K}]} {\theta[\vec{\Omega}_1
z+\vec{\phi}(m)+\vec{\eta}_-+\vec{K}]}+[\gamma+\varphi(0)].
\end{equation}
\end{Proposition}

\section{ Finite genus solutions to the lpKP}\label{sec-5}

Let $\gamma=\gamma_1,\,\gamma_2,\,\gamma_3$ be distinct and non-zero.
We can apply the same theory we developed in Section \ref{sec-4} to the three corresponding cases, respectively.
The resulting integrable maps
$S_{\gamma_1},\,S_{\gamma_2},\,S_{\gamma_3}$ commute in pairs since
they share the same integrals $E_1,\cdots,E_N$ (see Appendix in \cite{8-CaoZ-JPA-2012}).
By iteration we have discrete flows
$S_{\gamma_1}^{m_1},\,S_{\gamma_2}^{m_2},\,S_{\gamma_3}^{m_3}$,
and hence well-defined functions from any starting point
\begin{subequations}\label{5.1}
\begin{align}
&(p(m_1,m_2,m_3),q(m_1,m_2,m_3))=S_{\gamma_1}^{m_1}S_{\gamma_2}^{m_2}S_{\gamma_3}^{m_3}(p_0,q_0),
\label{5.1a}\\
&(u(m_1,m_2,m_3),v(m_1,m_2,m_3))=(-<p,p>,<q,q>)\vert_{(m_1,m_2,m_3)}.\label{5.1b}
\end{align}
\end{subequations}
Define $a=u$, and let $b$ take $\tilde{v}=T_1v,\,\bar{v}=T_2v,\,\hat{v}=T_3v$, respectively.
By the commutativity of the flows, one can present the functions given by equation (\ref{5.1a}) in three ways, respectively  as
\begin{equation}\label{5.2}
(p(m_k),q(m_k))=S_{\gamma_k}^{m_k}(p_0^{(k)},q_0^{(k)}),\quad(k=1,2,3).
\end{equation}
Thus, from equation (\ref{4.1}) in the three special cases, the $j$-th
component satisfies three equations simultaneously with $\lambda=\alpha_j$,
\begin{equation}\label{5.3}
T_k{p_j\choose
q_j}=(\alpha_j-\gamma_k)^{-1/2}D^{(\gamma_k)}(\alpha_j;u,T_kv){p_j\choose
q_j},\quad(k=1,2,3).
\end{equation}
Introducing
\begin{equation}\label{5.4}
\chi=(\alpha_j-\gamma_1)^{m_1/2}(\alpha_j-\gamma_2)^{m_2/2}(\alpha_j-\gamma_3)^{m_3/2}{p_j\choose
q_j},
\end{equation}
we then have
\begin{equation}\label{5.5}
T_k\chi=D^{(\gamma_k)}(\alpha_j;u,T_kv)\chi,\quad(k=1,2,3).
\end{equation}
In other words, the overdetermined system of equations (\ref{2.2}a-c) has
a compatible solution $\chi$ for the parameter $\lambda=\alpha_j$.
Now, for the lpKP equation (\ref{1.11}), recalling Proposition \ref{P-2.2}, we arrive at the following.
\begin{Proposition}\label{P-5.1}
The lpKP equation (\ref{1.11}),
$\Xi^{(0,3)}=0$, has a special solution
\begin{align}
W(m_1,m_2,m_3)=&2\partial_z\vert_{z=0}\mathrm{log}
\frac{\theta[z\vec{\Omega}_1+\vec{\phi}(m_1,m_2,m_3)+\vec{\eta}_-+\vec{K}]}
{\theta[z\vec{\Omega}_1+\vec{\phi}(0,0,0)+\vec{\eta}_-+\vec{K}]}\nonumber\\
&-2\sum_{s=1}^3m_s[\gamma_s+\varphi_s(0)]+W(0,0,0),\label{5.6}
\end{align}
where
\begin{equation}\label{5.7}
\vec{\phi}(m_1,m_2,m_3)=\sum_{s=1}^3m_s\vec{\Omega}_{\gamma_s}+\vec{\phi}(0,0,0),
\end{equation}
and $\varphi_s(z)$ is defined by equation (\ref{4.13}) in the case of
$\gamma=\gamma_s$.
\end{Proposition}

\begin{proof}
We have $(k=1,2,3)$
\[
T_kW-W=2\partial_z\vert_{z=0}\text{log}\frac{\theta[z\vec{\Omega}_1+T_k\vec{\phi}(m_1,m_2,m_3)+\vec{\eta}_-+\vec{K}]}
{\theta[z\vec{\Omega}_1+\vec{\phi}(m_1,m_2,m_3)+\vec{\eta}_-+\vec{K}]}-2[\gamma_k+\varphi_k(0)].
\]
It is equal to $-2u(T_kv)$ by equation (\ref{4.18}). Thus $W$ solves
(\ref{2.3}a-c) simultaneously. According to Proposition \ref{P-2.2}, $W$ solves
equation (\ref{1.11}).
\end{proof}

\section{ Solutions of other equations}\label{sec-6}
In solving  pKP equation  $\Xi^{(j,k)}=0$ that contains at least one continuous argument  $x$,
we will derive an explicit analytic expression for  $uv$,
which is similar to  (\ref{4.18}) and also  meets the auxiliary equation (\ref{2.3d}).
This can be done on the Liouville integrable platform as well, like in the discrete case.
We list the main steps as follows.

Consider $(p(x),q(x))=g_{H_1}^x(p_0,q_0)$. Hence $(u(x),v(x))=(-<p,p>,<q,q>)$ provide the finite genus potential functions.
For the ZS-AKNS equation (\ref{1.4}) with these potential functions, the solution space
$\mathcal{E}_{\lambda}$ is invariant under the action of $L(\lambda)$ due to the commutativity relation (\ref{3.3}).
The linear operator $L(\lambda)$  has eigenvalues $\pm H(\lambda)$,
with associated eigenvectors $h_{\pm}$ in $\mathcal{E}_{\lambda}$, satisfying
\begin{subequations}\label{6.1}
\begin{align}
&L(\lambda)h_{\pm}(x,\lambda)=\pm H(\lambda)h_{\pm}(x,\lambda),\label{6.1a}\\
&\partial_xh_{\pm}(x,\lambda)=U_1(\lambda;u(x),v(x))h_{\pm}(x,\lambda).\label{6.1b}
\end{align}\end{subequations}
Let $M(x,\lambda)$  be basic solution matrix of equation (\ref{6.1b}). The
eigenvectors are uniquely determined under the normalized condition
$h_{\pm}^{(2)}(0,\lambda)=1$ and can be expressed as
\begin{subequations}\label{6.2}
\begin{align}
&h_{\pm}(x,\lambda)={h_{\pm}^{(1)}(x,\lambda)\choose
h_{\pm}^{(2)}(x,\lambda)}=M(x,\lambda){c^{\pm}_{\lambda}\choose 1},\label{6.2a}\\
&c^{\pm}_{\lambda}=\frac{L^{11}(0,\lambda)\pm
H(\lambda)}{L^{21}(0,\lambda)}.\label{6.2b}
\end{align}\end{subequations}
Two meromorphic functions $\mathfrak{h}^{(\kappa)}(x,\mathfrak{p})$,
$\kappa=1,2$, are defined in $\mathcal{R}-\{\infty_+,\,\infty_-\}$
by
\begin{equation*}
\mathfrak{h}^{(\kappa)}(x,\mathfrak{p}(\lambda))=h_+^{(\kappa)}(x,\lambda),\quad
\mathfrak{h}^{(\kappa)}(x,\tau\mathfrak{p}(\lambda))=h_-^{(\kappa)}(x,\lambda).
\end{equation*}
A formula of Dubrovin-Novikov's type is derived from the
commutativity relation (\ref{3.3}). It is used to calculate the divisor of
$\mathfrak{h}^{(2)}(x,\mathfrak{p})$, which is equal to
$\sum_{j=1}^g[\mathfrak{p}(\nu_j(x))-\mathfrak{p}(\nu_j(0))]$. By
equations (\ref{3.11b}) and (\ref{3.15b}), we have
\begin{equation}\label{6.3}
\vec{\phi}(x)=\mathscr{A}\Bigl(\sum_{j=1}^g\mathfrak{p}(\nu_j(x))\Bigr)\equiv
x\vec{\Omega}_1+\vec{\phi}(0),\quad(\text{mod}\,\mathscr{T}).
\end{equation}
On the two-sheeted Riemann surface  $\mathcal {R}$, an Abel
differential, $\omega^{(1)}[\infty_-,\infty_+]$, of the third kind
is constructed, having only poles at $\infty_-,\,\infty_+$ with
\begin{equation}\label{6.4}
\omega^{(1)}[\infty_-,\infty_+]=\begin{cases}
[-z^{-2}-a^{(1)}(z)]\text{d}z,&\text{near $\infty_+$},\\
[+z^{-2}+a^{(1)}(z)]\text{d}z,&\text{near $\infty_-$},
\end{cases}
\end{equation}
where $a^{(1)}(z)$ is holomorphic near $z\sim 0$.
Without loss of generality, it can be arranged to satisfy the condition
\begin{equation}\label{6.5}
\int_{a_j}\omega^{(1)}=0,\quad\int_{b_j}\omega^{(1)}=-4\pi
\text{i}\Omega_1^j,\quad(1\leq j\leq g),
\end{equation}
where $\vec{\Omega}_1=(\Omega_1^1,\cdots,\Omega_1^g)^T$. Actually,
by adding a linear combination of holomorphic differentials
$\omega_1,\cdots,\omega_g$  to $\omega^{(1)}$, we can make the
former formula in equation (\ref{6.5}) valid. The latter is a corollary of
the former, which can be verified by using the canonical
representation of the Riemann surface  $\mathcal{R}$ \cite{10,20}. The
form of local expressions (\ref{6.4}) is invariant with adjusted
$a^{(1)}(z)$. We adopt the same symbol, for short. Through a usual
analysis we reconstruct \cite{4,20}
\begin{equation}\label{6.6}
\mathfrak{h}^{(2)}(x,\mathfrak{p})=c^{(2)}(x)\frac{\theta[-\mathscr{A}(\mathfrak{p})+\vec{\phi}(x)+\vec{K}]}
{\theta[-\mathscr{A}(\mathfrak{p})+\vec{\phi}(0)+\vec{K}]}
\cdot\text{exp}\Big(\frac{x}{2}\int_{\mathfrak{p}_0}^{\mathfrak{p}}\omega^{(1)}[\infty_-,\infty_+]\Big),
\end{equation}
where $c^{(2)}$ is independent of  $\mathfrak{p}\in\mathcal{R}$.
Equations (\ref{6.5}) are used to cancel the extra factors caused by
the uncertain linear combination of the contours
$a_1,\cdots,a_g,b_1,\cdots,b_g$ in the integration route  from the
point $\mathfrak{p}_0$ to $\mathfrak{p}$, both in
$\mathscr{A}(\mathfrak{p})$ and in the integral of $\omega^{(1)}$.

By equation (\ref{3.9}), near $\infty_-$ we have ($z=\lambda^{-1}\sim 0$)
\begin{equation*}
-\mathscr{A}(\mathfrak{p})=\vec{\eta}_-+\vec{\Omega}_1z+O(z^2),\quad\vec{\eta}_-=\int_{\infty_-}^{\mathfrak{p}_0}
\vec{\omega}.
\end{equation*}
Exerting action  $\partial_z\partial_x\text{log}$ on equation (\ref{6.6}),
we obtain $\partial_z\partial_x\text{log}h_-^{(2)}$, which is equal to
\begin{equation}\label{6.7}
\partial_z\partial_x\text{log}\theta[\vec{\Omega}_1z+O(z^2)+\vec{\phi}(x)+\vec{\eta}_-
+\vec{K}]+\frac{1}{2}[z^{-2}+a^{(1)}(z)].
\end{equation}
On the other hand, by equation (\ref{6.1a}) we estimate
\begin{equation}\label{6.8}
v\frac{h_-^{(1)}}{h_-^{(2)}}=v\frac{L^{11}(\lambda)-H(\lambda)}{L^{21}(\lambda)}=-uv\lambda^{-1}+O(\lambda^{-2}),
\end{equation}
where the following estimations are employed,
\begin{align*}
&L^{11}(\lambda)=\frac{1}{2}+<p,q>z+<Ap,q>z^2+O(z^3), \\
&L^{21}(\lambda)=<q,q>z+O(z^2), \\
&H(\lambda)=\frac{1}{2}-2H_0z-2H_1z^2+O(z^3).
\end{align*}
From equation (\ref{6.1b}) and the estimation (\ref{6.8}), we have
\begin{subequations}\label{6.9}
\begin{align}
&\partial_x\text{log}h_-^{(2)}=-\frac{\lambda}{2}+v\frac{h_-^{(1)}}{h_-^{(2)}}
=-\frac{1}{2}z^{-1}-uvz+O(z^2),\label{6.9a}\\
&\partial_z\partial_x\text{log}h_-^{(2)}=z^{-2}/2-uv+O(z).\label{6.9b}
\end{align}
\end{subequations}
Then, equating equation (\ref{6.7}) with (\ref{6.9b})  to cancel the singular term $z^{-2}/2$,
we obtain the following.
\begin{Proposition}\label{P-6.1}
Let $(u,v)$  be finite genus potential functions for equation (\ref{1.4}). Then
\begin{equation}\label{6.10}
-2uv=2\partial_z\vert_{z=0}\partial_x\mathrm{log}\,\theta[\vec{\Omega}_1z
+\vec{\phi}(x)+\vec{\eta}_-+\vec{K}]+a^{(1)}(0).
\end{equation}
\end{Proposition}

Next, we can recover $W$ for the pKP equations with continuous arguments.
In order to solve  $\Xi^{(1,2)}=0$, we consider the integrable maps
$S_{\gamma_1,}\,S_{\gamma_2}$ and $g_{H_1}^x$,
which commute in pairs since they share the same integrals $\{E_j\}$ (cf.\cite{8-CaoZ-JPA-2012}).
Well-defined functions are constructed as
\begin{subequations}\label{6.11}
\begin{align}
&(p(x,m_1,m_2),q(x,m_1,m_2))=g_{H_1}^xS_{\gamma_1}^{m_1}S_{\gamma_2}^{m_2}(p_0,q_0),\label{6.11a}\\
&(u(x,m_1,m_2),v(x,m_1,m_2))=(-<p,p>,<q,q>)\vert_{(x,m_1,m_2)}.\label{6.11b}
\end{align}
\end{subequations}
By the commutativity of the flows, the functions in equation (\ref{6.11a})
can be presented in three ways, respectively, as
\begin{subequations}\label{6.12}
\begin{align}
&(p(x),q(x))=g_{H_1}^x(p_0^{\prime},q_0^{\prime}),\label{6.12a}\\
&(p(m_k),q(m_k))=S_{\gamma_k}^{m_k}(p_0^{(k)},q_0^{(k)}),\quad(k=1,2).\label{6.12b}
\end{align}
\end{subequations}
Thus the $j$-th component satisfies three equations simultaneously
with $\lambda=\alpha_j$,
\begin{subequations}\label{6.13}
\begin{align}
&\partial_x{p_j\choose q_j}=U_1(\alpha_j;u,v){p_j\choose q_j,},\label{6.13a}\\
&T_k{p_j\choose
q_j}=(\alpha_j-\gamma_k)^{-1/2}D^{(\gamma_k)}(\alpha_j;u,T_kv){p_j\choose
q_j},\quad(k=1,2).\label{6.13b}
\end{align}
\end{subequations}
Introducing
$\chi=(\alpha_j-\gamma_1)^{m_1/2}(\alpha_j-\gamma_2)^{m_2/2}(p_j\,q_j)^T$,
we have
\begin{subequations}\label{6.14}
\begin{align}
&\partial_x\chi=U_1(\alpha_j;u,v)\chi,\label{6.14a}\\
&T_k\chi=D^{(\gamma_k)}(\alpha_j;u,T_kv)\chi,\quad(k=1,2).\label{6.14b}
\end{align}
\end{subequations}
Thus equations (\ref{1.4}) and (\ref{2.2}a,b) have a compatible solution  $\chi$
for the parameter $\lambda=\alpha_j$.

\begin{Proposition}\label{P-6.2}
The semi-discrete pKP equation (\ref{1.12}), i.e. $\Xi^{(1,2)}=0$, has a solution
\begin{align}
W(x,m_1,m_2)=&2\partial_z\vert_{z=0}\mathrm{log}
\frac{\theta[z\vec{\Omega}_1+\vec{\phi}(x,m_1,m_2)+\vec{\eta}_-+\vec{K}]}
{\theta[z\vec{\Omega}_1+\vec{\phi}(0,0,0)+\vec{\eta}_-+\vec{K}]}\nonumber\\
&-2\sum_{s=1}^2m_s[\gamma_s+\varphi_s(0)]+a^{(1)}(0)x+W(0,0,0),\label{6.15}
\end{align}
where
$\vec{\phi}(x,m_1,m_2)=x\vec{\Omega}_1+\sum_{s=1}^2m_s\vec{\Omega}_{\gamma_s}+\vec{\phi}(0,0,0)$,
 $\varphi_s(z)$ defined by equation (\ref{4.13}) with
$\gamma=\gamma_s$, and $a^{(1)}(z)$ given by equation (\ref{6.4}).
\end{Proposition}

\begin{proof}
From \eqref{6.15} we have ($k=1,2$)
\begin{align*}
&\partial_xW=2\partial_z\vert_{z=0}\partial_x\text{log}
\theta[z\vec{\Omega}_1+\vec{\phi}(x,m_1,m_2)+\vec{\eta}_-+\vec{K}]+a^{(1)}(0), \\
&T_kW-W=2\partial_z\vert_{z=0}\text{log}\frac{\theta[z\vec{\Omega}_1+T_k\vec{\phi}(x,m_1,m_2)+\vec{\eta}_-+\vec{K}]}
{\theta[z\vec{\Omega}_1+\vec{\phi}(x,m_1,m_2)+\vec{\eta}_-+\vec{K}]}-2[\gamma_k+\varphi_k(0)],
\end{align*}
which are equal to $-2uv$ and $-2u(T_kv)$ according to  equations (\ref{6.10}) and (\ref{4.18}), respectively.
Recalling Proposition \ref{P-2.3},  $W$ solves equation (\ref{1.12}).
\end{proof}

By similar analysis we have the following.

\begin{Proposition}\label{P-6.3}
The semi-discrete pKP equation (\ref{1.13}),  i.e. $\Xi^{(2,1)}=0$, has a solution
\begin{align}
W(x,y,m_1)=&2\partial_z\vert_{z=0}\mathrm{log}
\frac{\theta[z\vec{\Omega}_1+\vec{\phi}(x,y,m_1)+\vec{\eta}_-+\vec{K}]}
{\theta[z\vec{\Omega}_1+\vec{\phi}(0,0,0)+\vec{\eta}_-+\vec{K}]}\nonumber \\
&-2m_1[\gamma_1+\varphi_1(0)]+a^{(1)}(0)x+W(0,0,0),\label{6.16}
\end{align}
where
$\vec{\phi}(x,y,m_1)=x\vec{\Omega}_1+y\vec{\Omega}_2+m_1\vec{\Omega}_{\gamma_1}+\vec{\phi}(0,0,0)$.
\end{Proposition}

\begin{Proposition}\label{P-6.4}
The pKP equation (\ref{1.3}), i.e. $\Xi^{(3,0)}=0$, is solved by
\begin{align}
W(x,y,t)=&2\partial_z\vert_{z=0}\mathrm{log}\frac{\theta[z\vec{\Omega}_1+\vec{\phi}(x,y,t)+\vec{\eta}_-+\vec{K}]}
{\theta[z\vec{\Omega}_1+\vec{\phi}(0,0,0)+\vec{\eta}_-+\vec{K}]}\nonumber\\
&+a^{(1)}(0)x+W(0,0,0),\label{6.17}
\end{align}
where
$\vec{\phi}(x,y,t)=x\vec{\Omega}_1+y\vec{\Omega}_2+t\vec{\Omega}_3+\vec{\phi}(0,0,0)$.
\end{Proposition}

\section{Concluding remarks}\label{sec-7}

In this paper we have shown that the lpKP equation, semi-discrete pKP equations
and  continuous pKP equation
can be derived as compatibilities of Lax triads that originate from the ZS-AKNS spectral problems.
The approach to constructing finite genus solutions for 2D lattice equations \cite{7-CaoX-JPA-2012,8-CaoZ-JPA-2012}
was extended to 3D cases. As a result, we obtained finite genus solutions for
the discrete, semi-discrete and continuous pKP equations.
Note that these solutions are different from the elliptic solitons that are genus-one solutions
obtained by Nijhoff, et al. in \cite{NijA-IMRN-2010,YooN-JMP-2013}.

In deriving those pKP equations, we employed the auxiliary relations \eqref{2.3}.
We note that usually $W$ in \eqref{2.3} can not be exactly solved out for all arbitrarily given $(u,v)$;
therefore it is hard to say when Lax triads provide strict integrability for 3D equations in some cases (cf.\cite{LevRS-CJP-1994}).
However, as for the case of finite genus solutions,  since the finite-dimensional integrable flows $g_{H_j}^{\tau_j}$ and $S_{\gamma_k}^{m_k}$  share same Liouville integrals, same Lax matrix  and  same algebraic curve,
it enables us to treat \eqref{2.3} on the same Liouville platform
and obtain explicit expressions for $W$ from \eqref{2.3} by algebro-geometric integration.

The lpKP equation is one of the five octahedron-type equations
with 4D consistency \cite{AdlBS-IMRN-2012}.
We believe our approach can be extended to other octahedron-type integrable equations.
This will be a part of our future work.

\vskip 8pt

\subsection*{Acknowledgments}

The authors are grateful to the referee for the invaluable comments.
This work is supported by the National Natural Science Foundation of
China (grant nos 10971200, 11501521, 11426206, 11631007 and 11875040 for the
three authors, respectively).

\vskip 8pt

\section*{Appendices}
\appendix
\section{An heuristic deduction of Ansatz (\ref{2.14})}\label{A-1}

 The Abel-Jacobi variable $\vec{\phi}$ in the Jacobian
variety $J(\mathcal{R})$ provides a favorable window to observe the
evolution of the discrete symplectic flow $S_{\gamma_k}^{m_k}$ as
well as the Hamiltonian flow $g_{H_j}^{\tau_j}$. The discrete
velocity $\vec{\Omega}_{\gamma_k}$ of $\vec{\phi}$ is given by
equation (\ref{4.10c}), while the continuous velocity $\vec{\Omega}_j$ is
calculated by equation (\ref{3.9}) and (\ref{3.15a}). They are bridged by the
normalized basis $\vec{\omega}$ of holomorphic differentials. Let
the parameter $\gamma=\gamma_k$ tend to be infinity in the way as
$\gamma_k=-1/\varepsilon_k$, with $\varepsilon_k=c_k\varepsilon$,
$\varepsilon\rightarrow 0$. By the local expression of $\vec{\omega}$ near $\infty_+$  given by equation (\ref{3.9}),
we have
\begin{equation*}
\vec{\Omega}_{\gamma_k}=\int_{\mathfrak{p}(\gamma_k)}^{\infty_+}\vec{\omega}=
\varepsilon_k\vec{\Omega}_1-\frac{\varepsilon_k^2}{2}\vec{\Omega}_2+\frac{\varepsilon_k^3}{3}\vec{\Omega}_3
+O(\varepsilon^4).
\end{equation*}
Substituting this into equation (\ref{5.7}), we obtain
\begin{equation*}
\vec{\phi}-\vec{\phi}_0=
\sum_{k=1}^3m_k\big(\varepsilon_k\vec{\Omega}_1
-\frac{\varepsilon_k^2}{2}\vec{\Omega}_2+\frac{\varepsilon_k^3}{3}\vec{\Omega}_3\big)
+O(\varepsilon^4).
\end{equation*}
On the other hand, by equation (\ref{6.17}), the 3D continuous evolution of $\vec{\phi}$ reads
\begin{equation*}
\vec{\phi}-\vec{\phi}_0=(x-x_0)\vec{\Omega}_1+(y-y_0)\vec{\Omega}_2+(t-t_0)\vec{\Omega}_3.
\end{equation*}
Thus, up to $O(\varepsilon^4)$, we have
\begin{equation*}\begin{split}
&x-x_0=\sum_{s=1}^3m_k\varepsilon_k,\quad y-y_0=-\sum_{s=1}^3m_k\frac{\varepsilon_k^2}{2},\quad t-t_0=\sum_{s=1}^3m_k\frac{\varepsilon_k^3}{3},\\
&T_kx=x+\varepsilon_k,\quad T_ky=y-\frac{\varepsilon_k^2}{2},\quad
T_kt=t+\frac{\varepsilon_k^3}{3}.
\end{split}\end{equation*}
By substituting them into  $T_kW=W(T_kx,T_ky,T_kt)$, we obtain Ansatz (\ref{2.14}).

\section{Continuum limit of the lNLS}\label{A-2}

The lNLS equation (\ref{2.4}), i.e. $\Xi^{(0,2)}=0$, is first obtained by
Konopelchenko \cite{15}. It is solved in \cite{8-CaoZ-JPA-2012}. At first glance, its
relation with the NLS equation (\ref{2.9}), $\Xi^{(2,0)}=0$, is not clear.
It turns out that there is a transformation of Nijhoff's type,
\begin{equation*}
u=(-\gamma_1)^{m_1}(-\gamma_2)^{m_2}u^{\prime},\quad
v=(-\gamma_1)^{-m_1}(-\gamma_2)^{-m_2}v^{\prime},
\end{equation*}
which reduces the lNLS equation into an equation of
$(u^{\prime},v^{\prime})$,
\begin{equation*}\begin{split}
&(\Xi^{\prime})^{(0,2)}_1\equiv(\gamma_1\gamma_2)^{-1}(\gamma_1\tilde{u}^{\prime}
-\gamma_2\bar{u}^{\prime})u^{\prime}\bar{\tilde{v}}^{\prime}
+\gamma_1(\tilde{u}^{\prime}-u^{\prime})-\gamma_2(\bar{u}^{\prime}-u^{\prime})=0,\\
&(\Xi^{\prime})^{(0,2)}_2\equiv(\gamma_1\gamma_2)^{-1}(\gamma_2\tilde{v}^{\prime}
-\gamma_1\bar{v}^{\prime})u^{\prime}\bar{\tilde{v}}^{\prime}
+\gamma_1(\tilde{\bar{v}}^{\prime}-\bar{v}^{\prime})-\gamma_2(\bar{\tilde{v}}^{\prime}-\tilde{v}^{\prime})=0.
\end{split}\end{equation*}
Let $-\gamma_k^{-1}=\varepsilon_k=c_k\varepsilon$, ( $k=1,2$ ),
where $c_1,c_2$ are distinct non-zero constants. For any smooth
functions $u^{\prime}(x,y)$, $v^{\prime}(x,y)$, define
\begin{equation*}\begin{split}
&\tilde{u}^{\prime}=u^{\prime}(x+\varepsilon_1,y-\varepsilon_1^2/2),\quad
\bar{u}^{\prime}=u^{\prime}(x+\varepsilon_2,y-\varepsilon_2^2/2),\\
&\bar{\tilde{u}}^{\prime}=u^{\prime}(x+\varepsilon_1+\varepsilon_2,y-\varepsilon_1^2/2-\varepsilon_2^2/2),
\end{split}\end{equation*}
and similar expressions for
$\tilde{v}^{\prime},\bar{v}^{\prime},\bar{\tilde{v}}^{\prime}$.
Then, as $\varepsilon\sim 0$, we have the following Taylor expansion
which confirms that the continuum limit of the lNLS is NLS up to a Nijhoff's type transformation,
\begin{equation*}
(\Xi^{\prime})^{(0,2)}={u_y^{\prime}-u_{xx}^{\prime}+2(u^{\prime})^2v^{\prime}
 \choose
 v_y^{\prime}+v_{xx}^{\prime}-2u^{\prime}(v^{\prime})^2}\frac{c_1-c_2}{2}\varepsilon+O(\varepsilon^2).
\end{equation*}

Similarly, the semi-discrete  NLS equation (\ref{2.11}),
$\Xi^{(1,1)}=0$, is transformed as
\begin{equation*}
(\Xi^{\prime})^{(1,1)}\equiv{u^{\prime}_x+\gamma_1(\tilde{u}^{\prime}-u^{\prime})
-\gamma_1^{-1}(u^{\prime})^2\tilde{v}^{\prime}
\choose
\tilde{v}^{\prime}_x+\gamma_1(\tilde{v}^{\prime}-v^{\prime})+\gamma_1^{-1}u^{\prime}(\tilde{v}^{\prime})^2}=0,
\end{equation*}
by the transformation of Nijhoff's type,
$u=(-\gamma_1)^{m_1}u^{\prime},\,v=(-\gamma_1)^{-m_1}v^{\prime}$ .
Let
\begin{equation*}\begin{split}
\tilde{u}^{\prime}=u^{\prime}(x+\varepsilon_1,y-\varepsilon_1^2/2),\quad
\tilde{v}^{\prime}=v^{\prime}(x+\varepsilon_1,y-\varepsilon_1^2/2).
\end{split}\end{equation*}
Then, as  $\varepsilon\sim 0$, we have the following Taylor
expansion, which confirms that the continuum limit of the
time-discrete NLS equation (\ref{2.11}) is the NLS equation up to a Nijhoff's type transformation,
\begin{equation*}
(\Xi^{\prime})^{(1,1)}={u_y^{\prime}-u_{xx}^{\prime}+2(u^{\prime})^2v^{\prime}
 \choose
 v_y^{\prime}+v_{xx}^{\prime}-2u^{\prime}(v^{\prime})^2}\frac{c_1}{2}\varepsilon+O(\varepsilon^2).
\end{equation*}

\end{document}